\documentclass[11pt,a4paper,reqno]{amsproc}

\IfFileExists{newheader.sty}{\input{newheader.sty}}{\input{../newheader.sty}}

\title{Enhanced area law in the Widom--Sobolev formula for the free Dirac operator in arbitrary dimension}

\author[L.\ Bollmann]{Leon Bollmann}
\address[L.\ Bollmann]{Mathematisches Institut,
  Ludwig-Maximilians-Universit\"at M\"unchen,
  Theresienstra\ss{e} 39,
  80333 M\"unchen, Germany}
\email{bollmann@math.lmu.de}

\author[P.\ M\"uller]{Peter M\"uller}
\address[P.\ M\"uller]{Mathematisches Institut,
 Ludwig-Maximilians-Universit\"at M\"unchen,
  Theresienstra\ss{e} 39,
  80333 M\"unchen, Germany}
\email{mueller@lmu.de}

\begin{document}

\begin{abstract}
	We prove a logarithmically enhanced area law for all R\'enyi entanglement entropies of the ground state 
	of a free gas of relativistic Dirac fermions. Such asymptotics occur in any dimension if the modulus of 
	the Fermi energy is larger than the mass of the particles and in the massless case at Fermi energy zero 
	in one space dimension. In all other cases of mass, Fermi energy and dimension, the entanglement entropy 
	grows no faster than the area of the involved spatial region. The result is established for a general class of test functions which includes the ones corresponding to R\'enyi entropies and relies on a recently proved extension of the Widom--Sobolev formula to matrix-valued symbols by the authors. 
\end{abstract}

\maketitle

\section{Introduction}

The last decades have witnessed considerable progress in our understanding of Szeg\H{o} asymptotics for 
traces of truncated Wiener--Hopf operators in multi-dimensional space. A prime achievement is the proof of Widom's 
conjecture by Sobolev \cite{sobolevlong} with various extensions to more general situations \cite{sobolevschatten, sobolevpw, sobolevc2}.
Such methods have proved to be a very powerful mathematical tool in the study of 
entanglement properties of ground states of non-relativistic quasi-free Fermi gases. Most notably, we mention the proof 
of a logarithmically enhanced area law for the R\'enyi entropies of the reduced ground state of a free 
Fermi gas \cite{LeschkeSobolevSpitzerLaplace}, which relies on the Widom--Sobolev formula. Generalisations to positive temperatures and towards the inclusion of electric and magnetic background fields were treated, e.g., in \cite{LeschkeSobolevSpitzerLaplace, PhysRevLett.113.150404, ElgartPasturShcherbina2016, LeschkeSobolevSpitzerMultiScale, Dietlein18, MuPaSc19, MuSc20, LeschkeSobolevSpitzerMagnetic, Pfeiffer21, LeschkeSobolevSpitzerTemperature, 10.1063/5.0135006, PfSp22, PfSp23}.

It is therefore a natural question whether there also exists a logarithmically enhanced area law for relativistic quasi-free Fermi gases, primarily for the case of a free Dirac operator as the single-particle Hamiltonian. The recent paper \cite{FinsterSob}, which is also devoted to a mathematical study of entanglement entropies for a free Fermi gas governed by the Dirac operator, concentrates on the case of Fermi energy $E_{F}=0$ and finds a strict area law in dimension $d=3$, even in the massless case, $m=0$, where energy zero does not lie in a gap in the spectrum.
The Dirac equation in globally hyperbolic spacetimes with applications to black-hole entropy is discussed in \cite{FinsterI}.
Entanglement signatures of Dirac systems are also a current topic of interest in the physics literature \cite{10.1088/1742-5468/aa668a, 10.1126/sciadv.aat5535, Crepel.Hackenbroich.2021}.

Our main focus in this paper is on situations which lead to a logarithmically enhanced area law for the free Dirac system. 
As in the case for the Laplacian \cite{LeschkeSobolevSpitzerLaplace}, we would like to trace the problem back to 
a suitable Widom--Sobolev formula which, however, in this case is required to hold for matrix-valued symbols. The derivation of such a formula has been performed in our previous paper \cite{BM-Widom} where we follow 
the path laid out by Sobolev \cite{sobolevlong, sobolevc2} and deal with the additional complications posed by matrix-valued symbols, most notably issues of non-commutativity.
We also refer to \cite{BM-Widom} for an account of the literature on Szeg\H{o} asymptotics for matrix-valued symbols. 

The main results of this paper are summarised in Theorem~\ref{Dirac}. 
These are logarithmically enhanced area laws for the trace of rather general functions -- including all R\'enyi 
entropy functions -- of a spatially truncated Fermi projection of the free Dirac operator. 
Such asymptotics occur in the cases 
$|E_{F}| > m$ in all space dimensions $d\in\N$ and $E_{F}=m=0$ in $d=1$. In the complementary cases, 
i.e.\ $|E_{F}| \le m \neq 0$ for any $d$ and $E_{F}=m=0$ in $d\ge 2$, the coefficient of the enhanced 
area law 
vanishes, and we perform some further analysis to show that the corrections grow at most of the order of the 
area. While such a behaviour is not at all surprising in the spectral-gap case $|E_{F}| \le m \neq 0$, 
the interesting case $E_{F}=m=0$ in $d\ge 2$ requires to control the point singularity of the symbol at the origin in 
the sense that it is too low dimensional. The paper \cite{FinsterSob} studies the coefficient of the arising area law in great 
detail. We have included our corresponding estimate as Lemma~\ref{Schatten-Estimate} in the present paper both for completeness and because we discovered it independently. This is the only overlap with \cite{FinsterSob}.
From what has been said above, the point singularity of the symbol at the origin is strong enough, however, to induce a logarithmic enhancement  in  $d=1$. 

The core of the proof of Theorem~\ref{Dirac} consists in an application of the 
Widom--Sobolev formula \cite[Thm.~4.22]{BM-Widom} for matrix-valued symbols. In order to achieve this, one has to cope with several difficulties that are not present in the corresponding non-relativistic case for the Laplacian \cite{LeschkeSobolevSpitzerLaplace}. Firstly, unboundedness of the Dirac operator from below prevents the relevant operators from being trace class, and a suitable cut-off of the Fermi sea at negative energies has to be introduced. In order for this to be a meaningful procedure, we consider a setting in which the leading-order coefficient of the asymptotics turns out to be even independent of the cut-off. Secondly, the presence of two spectral bands at positive and negative energies leads to complications in the geometric structure of the relevant regions in momentum space. For this reason we distinguish several different cases in the proof, depending on the relevant regions in momentum space. Whereas for the Laplacian \cite{LeschkeSobolevSpitzerLaplace} the relevant region in momentum space is the interior of the bounded Fermi ball, the relevant momenta for the Dirac operator also lie in the unbounded complement of a ball. 
In order to obtain well-defined, i.e.\ finite, expressions in such a situation, the support of the symbol needs to be compact, which is guaranteed by the cut-off.
Moreover, in the Laplacian case the symbol always vanishes on one side of the discontinuity in momentum space. 
Here, we need to cope with symbols that take on non-zero values on both sides of the discontinuity.
The decomposition of the symbol in \eqref{Dirac-Application-Function} is tailored to reduce the analysis to an application of 
the Widom--Sobolev formula in \cite{BM-Widom}.
We emphasise that this relies on the particular matrix structure of the Dirac operator.

In the next section we introduce all relevant notions in order to formulate the main result, 
Theorem~\ref{Dirac}, at the beginning of Section~3. The remaining part of Section~3 is devoted to its proof. 


\section{Notation and preliminaries}
\label{sec:prelim}

\subsection{Admissible domains}

\begin{defn}\label{Definition-Admissible-Domain}
Given a natural number $d\geq 2$, we call $\Omega\subset\R^d$ an \emph{admissible domain}, if it can be locally represented as the graph of a Lipschitz function. By this we mean that for every $x\in\R^d$ there exists a radius $r_x>0$, a Lipschitz function $\Phi_x:\R^{d-1}\rightarrow\R$ and a choice $\R^{d} \ni y=(y_1,y_2,\ldots,y_d)$ of Cartesian coordinates
 such that 
\begin{align}\label{Def-Local-Representation}
B_{r_{x}}(x)\cap\Omega = B_{r_{x}}(x)\cap\big\{y\in\R^d : y_d>\Phi_x(y_1,y_2,\cdots,y_{d-1})\big\},
\end{align}
where $B_{r_{x}}(x)$ denotes the open ball  in $\R^{d}$ of radius $r_x$ about $x$.

For $m\in\N\cup\{\infty\}$, we further call an admissible domain a \emph{piece-wise $C^m$-admissible domain}, if the functions $\Phi_x$ are not only Lipschitz but also piece-wise $C^m$-functions.

In the case $d=1$ we call $\Omega\subset\R$ a (piece--wise $C^m$-) admissible domain, if it can be locally represented as an open interval of the form  $\,]a,\infty[\,$, respectively $\,]-\infty,a[\,$, for arbitrary $a\in\R$. 
\end{defn}

\begin{rem}
	\label{Remark-Complement-Admissible}
	\begin{enumerate}[(a)]
	\item
		Given two piece-wise $C^m$-admissible domains $\Omega\subset\Omega'$ such that $\dist(\Omega,\partial\Omega')>0$, the domain $\Omega'\setminus\overline{\Omega}$ is also piece-wise $C^m$-admissible. Here, $\overline{\Omega}$ denotes the closure of $\Omega$.
	\item
	The boundary of an admissible domain has Lebesgue-measure zero. Therefore, the indicator function $1_{\Omega}$ of some admissible domain $\Omega$ is equivalent in Lebesgue sense with the indicator function $1_{\overline{\Omega}}$. In particular in the situation of part (a) of this remark, the indicator function $1_{\Omega'\setminus\Omega}$ can be considered as the indicator function of a piece-wise $C^m$-admissible domain.
	\end{enumerate}
\end{rem}


\subsection{The free Dirac operator}

The self-adjoint Hamiltonian of the Dirac equation in 
$d\in\N$ space dimensions is given by 
\begin{equation}
	\label{Dirac-op}
	D:= -\i\sum_{k=1}^d \alpha_k \partial_k +m\beta,
\end{equation}
see e.g. \cite{10.1007/bf02754212, 10.1063/1.1367331}. It is densely defined in the product Hilbert space $L^{2}(\R^{d}) \otimes \C^{n_{d}}$ of square-integrable
vector-valued functions with spinor dimension
$n_{d}:= 2^{\lfloor (d+1)/2\rfloor}$. Here, $\lfloor r\rfloor$ stands for the largest integer not 
exceeding $r\in\R$. Moreover,
$\i$ is the imaginary unit, $m\ge 0$ is the mass of the particle, $\partial_{k}$, $k=1,\ldots,d$,
denotes partial differentiation with respect to the $k$-th Cartesian coordinate and 
$\beta,\alpha_{1},\ldots, \alpha_{d} \in\C^{n_{d}\times n_{d}}$ are Dirac matrices 
which anti-commute pairwise and whose square gives the $n_{d}\times n_{d}$-unit matrix $\mathbb{1}_{n_{d}}$.

The Dirac matrices can be chosen as (cf.\ \cite[Appendix]{10.1063/1.1367331})
\begin{align}
\alpha_j:=\begin{pmatrix}
0 & \sigma_j \\
 \sigma_j^* & 0
\end{pmatrix}, \;\; j=1,\ldots,d, \;\;  \text{and} \;\;  \beta := \begin{pmatrix}
\mathbb{1}_{\frac{n_{d}}{2}} & 0 \\
 0 & -\mathbb{1}_{\frac{n_{d}}{2}}
\end{pmatrix},
\end{align}
where the $\frac{n_{d}}{2}\times \frac{n_{d}}{2}$-matrices $\sigma_{j}$ satisfy the anti-commutation relations
\begin{align}
\label{sigma-anticomm}
\sigma_j \sigma_k^* + \sigma_k \sigma_j^* = 2\delta_{jk}\mathbb{1}_{\frac{n_{d}}{2}}
\;\; \text{and} \;\; \sigma_j^* \sigma_k + \sigma_k^* \sigma_j = 2\delta_{jk}\mathbb{1}_{\frac{n_{d}}{2}}
\end{align}
for all $j,k=1,\ldots,d$.
In the case that $d$ is odd, the matrices $\sigma_j$ are Hermitian for all $j=1,\ldots,d$.
In \eqref{sigma-anticomm}, $*$ denotes the adjoint and $\delta_{jk}:=1$ if $j=k$ and $0$ otherwise the Kronecker delta. 
 
The Hamiltonian given in \eqref{Dirac-op} is unitarily equivalent via the Fourier transform to the operator of multiplication on $L^{2}(\R^{d}) \otimes \C^{n_{d}}$ with the matrix-valued symbol 
\begin{align}
\R^{d} \ni \xi = (\xi_{1},\ldots,\xi_{d}) \mapsto \mathcal{D}(\xi):= \sum_{k=1}^d \alpha_k \xi_k +m\beta.
\end{align}
Clearly, the symbol $\mathcal{D}$ is smooth, i.e.\ $\mathcal{D} \in C^{\infty}(\R^{d})$.  

We want to study the application of measurable functions to this operator. In order to do so it will be convenient to diagonalise the matrix $\mathcal{D}(\xi)$ for every $\xi\in\R^{d}$. The matrix $\mathcal{D}(\xi)$ has both $E(\xi):=\sqrt{m^2+\xi^2}$ and $-E(\xi)$ as eigenvalues of multiplicity $\frac{n_{d}}{2}$. Therefore it can be diagonalised as
\begin{align}
U(\xi)\mathcal{D}(\xi)U^{-1}(\xi)=E(\xi)\beta,
\end{align}
where $U(\xi)$ is a unitary matrix. We now want to calculate $a(\mathcal{D}(\xi))$ for some measurable function $a:\R\rightarrow\C$.
Using the diagonalisation above, we have
\begin{align}
\label{symbol-diag}
a\big(\mathcal{D}(\xi)\big)=a\big(U^{-1}(\xi)E(\xi)\beta U(\xi)\big)=U^{-1}(\xi)a\big(E(\xi)\beta\big)U(\xi).
\end{align}
The matrix $E(\xi)\beta$ is diagonal with entries $E(\xi)$ and $-E(\xi)$. We define two functions $a_{\pm}$ on momentum space by
\begin{equation}
\label{a-pm-def}
a_+(\xi):= a\big(E(\xi)\big), \qquad a_-(\xi):= a\big(-E(\xi)\big).
\end{equation}
Then, \eqref{symbol-diag} can be written as
\begin{align}\label{Dirac-Application-Function}
a\big(\mathcal{D}(\xi)\big)=U^{-1}(\xi)a\big(E(\xi)\beta\big)U(\xi)=&\frac{1}{2}U^{-1}(\xi)\Big((a_++a_-)\mathbb{1}_{n_{d}}+ (a_+-a_-)\beta\Big) U(\xi)\nonumber\\=&\frac{1}{2}(a_++a_-)(\xi)\mathbb{1}_{n_{d}}+\frac{1}{2}(a_+-a_-)(\xi)\frac{\mathcal{D}}{E}(\xi)
\end{align}
for every $\xi\in\R^{d}$, where we define $\frac{\mathcal{D}}{E}(0):=0$ in the case $m=0$. 

\begin{rem}\label{Remark-Application-Function}
\begin{enumerate}[(a)]
\item In order for the notation in \eqref{Dirac-Application-Function} to make sense, we identify the matrix  $\mathbb{1}_{n_{d}}$ with a constant matrix-valued symbol. The multiplication of a scalar-valued symbol with a matrix-valued symbol is defined point-wise by scalar multiplication on the tensor product. These notational conventions will be used extensively throughout the paper. 
\item If the function $a$ in \eqref{Dirac-Application-Function} is (essentially) bounded, the symbol $a(\mathcal{D})$ is also (essentially) bounded. Then, for every $L>0$ the operator $\Op_{L}(a(\mathcal{D}))$, whose action on some $u\in L^{2}(\R^{d}) \otimes \C^{n_{d}}$ at a point $x\in\R^d$ is given by 
\begin{align}
\label{Definition-OpL}
	\big(\Op_L(a(\mathcal{D}))u\big)(x):= \left(\frac{L}{2\pi}\right)^d \int_{\R^d} \!\int_{\R^d} \e^{\i L\xi(x-y)}\mathcal{D}(\xi)u(y)\d y\d\xi,
\end{align}
is a well-defined and bounded operator on $L^{2}(\R^{d})\otimes\C^n$. We refer to \cite{BM-Widom} for more details on pseudodifferential operators in this context. We then write the application of the function $a$ to the Hamiltonian $D$ as a pseudo-differential operator with matrix-valued symbol $a(\mathcal{D})$, i.e. $a(D)=\Op_{1}(a(\mathcal{D}))$. 
\item \label{Remark-Compact-Support} If the function $a$ in \eqref{Dirac-Application-Function} is compactly supported, the functions $a_{+}$ and $a_{-}$ are also compactly supported. Therefore, the symbol $a(\mathcal{D})$ is compactly supported in this case.
\item 
	If there exists $\tilde{a}: \R\rightarrow\C$ measurable such that its support satisfies 
	$\supp(a-\tilde{a}) \subseteq [-m,m]$, then, by virtue of \eqref{Dirac-Application-Function}, the equality 
	$\tilde{a}(\mathcal{D}) = a(\mathcal{D})$ holds Lebesgue-almost everywhere, in particular we have 
	$\Op_{L}\big(\tilde{a}(\mathcal{D})\big) = \Op_{L}\big(a(\mathcal{D})\big)$.	
\item If the function $a$ in \eqref{Dirac-Application-Function} is smooth on $\ran \pm E = \,]-\infty,-m] \cup [m, \infty[\,$ (at the boundaries to be understood as one-sided derivatives), 
then the functions $a_{+}$ and $a_{-}$ are smooth on $\R^{d}$. 
In this case, the symbol $a(\mathcal{D})$ is smooth if $m>0$ because 
this guarantees smoothness of the symbol $\frac{\mathcal{D}}{E}$ on $\R^{d}$. If $m=0$, the additional assumption $a(0)=0$ also implies smoothness of  $a(\mathcal{D})$ despite the discontinuity of the symbol $\frac{\mathcal{D}}{E}$ at the origin. 
\end{enumerate}
\end{rem}

For the applications in this paper the following smoothly truncated version of the Fermi projection is of particular interest.
Given a Fermi energy $E_F \in\R$ and an ultraviolet cut-off parameter $b\in [0,\infty[\,$, we define 
\begin{equation}\label{Definition-Chi-E_F}
\chi_{E_F}^{(b)}:\R\rightarrow [0,1],\quad x\mapsto \chi_{E_F}^{(b)}(x):=1_{\{y\in\R\,:\, y<E_F\}}(x)\varphi_{E_{F}}(x+b)
\end{equation} 
in terms of the monotone cut-off function
$\varphi_{E_{F}} :=\varphi\in C^\infty(\R)$ obeying $\varphi|_{[-|E_F|,\infty[}=1$ and $\varphi|_{]-\infty,-|E_F|-1]}=0$. Note that the function $\chi_{E_F}^{(b)}$ is bounded and compactly supported.


\subsection{Test functions and entanglement entropy}

The test functions we consider are of the following type.
\begin{ass}
	\label{H-Functions}
	Let $\gamma \in \,]0,1[\,$ and let $X: =\{x_1,x_2,\ldots,x_N\}\subset\R,N\in\N,$ be a finite collection of different points 
	on the real line. Let $U_j\subset\R$, $j\in\{1,\ldots , N\}$, be pairwise disjoint 
	neighbourhoods of the points $x_j\in X$. Given a function 
	$h\in C(\R) \cap C^2(\R \setminus X)$, we assume the existence of a constant $C>0$ such that for every  $k\in\{0,1,2\}$ the estimate 
	\begin{align}\label{H-Functions-Bound}
		\Big|\frac{\d^k}{\d x^k}\big[h-h(x_j)\big](x)\Big|\leq C |x-x_j|^{\gamma-k}, 
	\end{align}
	holds for every $x\in U_j\setminus\{x_j\}$ and every $j\in\{1,\ldots , N\}$. In particular, this 
	implies that $h$ is H\"older continuous at the points of $X$.
\end{ass}

A special case of these functions are the R\'{e}nyi entropy functions. For a given parameter $\alpha\in\, ]0,\infty[\,\setminus \{1\}$ the function $h_{\alpha}:\R\rightarrow [0,\log 2]$ is given by
\begin{align}
h_{\alpha}(t):=\frac{1}{1-\alpha}\log [t^\alpha+(1-t)^\alpha]
\end{align}
for $t\in [0,1]$ and by $h_{\alpha}(t):=0$ elsewhere. Here $\log$ denotes the natural logarithm. In the case $\alpha=1$ the von Neumann entropy function $h_1$ is given by 
\begin{align}
h_{1}(t):=\lim_{\alpha\rightarrow 1}h_{\alpha}(t)=-t\log t-(1-t)\log (1-t)
\end{align}
for $t\in\, ]0,1[\,$ and by $h_{1}(t):=0$ elsewhere. We note that $h_{\alpha}(0)=0=h_{\alpha}(1)$ for all $\alpha>0$.

We now formulate the main object studied in this paper. Let $\Lambda,\Lambda' \subset\R^{d}$ with $\Lambda \subset\Lambda'$ be bounded piece-wise $C^{1}$-admissible domains. Given a subset $\Omega\subseteq \R^{d}$ and $L>0$, we introduce the scaled subset 
$\Omega_{L}:= \{L x: x\in \Omega\}$. We are interested in the following trace
\begin{align}\label{Paper-Goal}
\tr_{L^2(\R^d)\otimes\C^{n_{d}}} 
			\Big[h\big(\mathbf{1}_{\Lambda_{L}}\chi_{E_F}^{(b)}(D)\mathbf{1}_{\Lambda_{L}}\big) 
			+ h\big(\mathbf{1}_{\Lambda'_{L}\setminus\Lambda_{L}}\chi_{E_F}^{(b)}(D)
				\mathbf{1}_{\Lambda'_{L}\setminus\Lambda_{L}}\big)
			- h\big(\mathbf{1}_{\Lambda'_{L}}\chi_{E_F}^{(b)}(D)\mathbf{1}_{\Lambda'_{L}}\big)\Big],
\end{align}
for fixed Fermi energy $E_F\in\R$, mass $m\geq 0$ and energy cut-off $b\ge 0$. Here,  
$\mathbf{1}_{\Omega} := 1_{\Omega} \otimes \mathbb{1}_{n_{d}}$ acts as the multiplication operator on 
$L^{2}(\R^{d}) \otimes \C^{n_{d}}$ by the corresponding matrix-valued indicator function. If the function $h$ is chosen to be one of the R\'{e}nyi entropy functions, the trace \eqref{Paper-Goal} is related to the relative local R\'{e}nyi entropy corresponding to the domains $\Lambda$ and $\Lambda'$.
 
As discussed in the last section, we can rewrite each of the three terms in \eqref{Paper-Goal} in terms of a pseudo-differential operator, i.e. 
\begin{align}\label{Rewrite-PDO}
\chi_{E_F}^{(b)}(D) =\Op_{1}\big(\chi_{E_F}^{(b)}(\mathcal{D})\big),
\end{align}
where $\Omega$ is either $\Lambda,\Lambda'$ or $\Lambda'\setminus \Lambda$. As all of these domains are bounded and the symbol $\chi_{E_F}^{(b)}(\mathcal{D})$ is compactly supported, cf.\ Remark \ref{Remark-Application-Function} (\ref{Remark-Compact-Support}), all the operators in \eqref{Paper-Goal} are trace-class, see \cite[Chap.\
11, Sect.\ 8, Thm.\ 11]{BirmanSolomjak}. As the symbol $\chi_{E_F}^{(b)}(\mathcal{D})$ only depends on momentum space, the operator $\mathbf{1}_{\Omega_{L}}\Op_1\big(\chi_{E_F}^{(b)}(\mathcal{D})\big)\mathbf{1}_{\Omega_{L}}$ is unitarily equivalent to the operator $\mathbf{1}_{\Omega}\Op_L\big(\chi_{E_F}^{(b)}(\mathcal{D})\big)\mathbf{1}_{\Omega}$ by the dilatation $U_L$ on $L^2(\R^d)\otimes \C^n$, where $(U_L u)(x):=L^{\frac{d}{2}} u(Lx)$ for all $u\in L^2(\R^d)\otimes \C^n$.
Therefore, we infer
\begin{align}\label{Rewrite-L-PDO}
 \tr_{L^2(\R^d)\otimes\C^{n_{d}}} 
			\Big[h\big(\mathbf{1}_{\Omega_{L}}\Op_{1}\big(\chi_{E_F}^{(b)}(\mathcal{D})\big)\mathbf{1}_{\Omega_{L}}\big)\Big]= \tr_{L^2(\R^d)\otimes\C^{n_{d}}} 
			\Big[h\big(\mathbf{1}_{\Omega}\Op_{L}\big(\chi_{E_F}^{(b)}(\mathcal{D})\big)\mathbf{1}_{\Omega}\big)\Big].
\end{align}


\section{Main theorem and its proof}
\label{sec:main}
We are now able to state our main theorem.

\begin{thm}
	\label{Dirac}
	Let $\Lambda\subset\Lambda'$ be bounded piece-wise $C^1$-admissible domains in $\R^{d}$ in the sense of 
	Definition \ref{Definition-Admissible-Domain} and such that $\dist(\Lambda,\partial\Lambda')>0$. 
	Consider the Dirac operator \eqref{Dirac-op} with mass $m\ge0$ and fix a Fermi energy $E_{F}\in\R$ 
	and an ultraviolet cut-off parameter $b \ge0$.
	Let $h\in C(\R)$ satisfy Assumption \ref{H-Functions} and $h(0)=0$. Then, the asymptotic trace formula
	\begin{multline}
		\label{Dirac-stat}
		\frac{1}{n_{d}}\tr_{L^2(\R^d)\otimes\C^{n_{d}}} 
			\Big[h\big(\mathbf{1}_{\Lambda_{L}}\chi_{E_F}^{(b)}(D)\mathbf{1}_{\Lambda_{L}}\big) 
			+ h\big(\mathbf{1}_{\Lambda'_{L}\setminus\Lambda_{L}}\chi_{E_F}^{(b)}(D)
				\mathbf{1}_{\Lambda'_{L}\setminus\Lambda_{L}}\big)
			- h\big(\mathbf{1}_{\Lambda'_{L}}\chi_{E_F}^{(b)}(D)\mathbf{1}_{\Lambda'_{L}}\big)\Big] \\ 
		= L^{d-1}\log L \  W(h,\Lambda,E_F,m) + o(L^{d-1}\log L)
	\end{multline}
	holds as $L\rightarrow\infty$. 
	The coefficient $W(h,\Lambda,E_F,m)$ is independent of the cut-off $b$, the 
	domain $\Lambda'$ and the spinor dimension $n_{d}$. Moreover:
	\begin{enumerate}[{\upshape(a)}]
	\item 
		If $|E_F|>m$, then the coefficient of this enhanced area law is given by
		\begin{align}
			\label{Coefficient-Enhanced}
			W(h,\Lambda,E_F,m):= \frac{\Phi(\Lambda,E_F,m)}{(2\pi)^2}\int_0^1 \frac{h(t)-h(1)t}{t(1-t)}\;\d t 
		\end{align}
		with the geometric factor 
		\begin{align}
			\label{geom-factor}
			\Phi(\Lambda,E_F,m):= \LEFTRIGHT\{.{\begin{array}{ll}
     		2|\partial\Lambda|, & \text{if} \ \ d=1, \\[1ex]
					\displaystyle
     			\frac{1}{(2\pi)^{d-1}}\int_{\partial\Lambda}\int_{\partial B_{p_F}} \!
					\vert n_{\partial \Lambda}(x)\cdot n_{\partial B_{p_F}}(\xi)\vert\, \d S(\xi) \d S(x), 
							& \text{if } 	\  d\geq 2,
   		 \end{array}}
		\end{align}
		where, in one dimension, $|\partial \Lambda|$ is the number of boundary points of $\Lambda$. In 
		dimensions $d\geq 2$, we write $B_{p_F} := B_{p_{F}}(0)$, where $p_F:=\sqrt{E_F^2-m^2}$ 
		is the relativistic Fermi momentum, and $n_{\partial\Lambda}$, resp.\ $n_{\partial B_{p_F}}$, denotes the vector field of exterior unit normals in $\R^{d}$ to $\partial \Lambda$, resp.\ $\partial B_{p_{F}}$. We write $\d S$ for integration with respect to the $(d-1)$-dimensional surface measure induced by Lebesgue measure in $\R^{d}$. 
	\item 
		If $|E_F|\leq m\neq 0$, then $W(h,\Lambda,E_F,m)=0$, and the next term in the asymptotic expansion is of order $O(L^{d-1})$ as $L\to\infty$. 
	\item 
		If $E_F=m=0$, then the behaviour depends on the dimension. If $d=1$ an enhanced area law holds with the same coefficient \eqref{Coefficient-Enhanced} and \eqref{geom-factor} as in {\upshape(a)}. If instead $d\geq 2$, the situation is as in {\upshape(b)}.
	\end{enumerate}
\end{thm}

\begin{rem}
	\label{Dirac-rem}
	\begin{enumerate}[(a)]		
	
\item We note that in $d=1$ the geometric factor $\Phi(\Lambda,E_F,m) =2|\Lambda|$ and subsequently the coefficient 
$W(h,\Lambda,E_F,m)$ are actually independent of the Fermi energy $E_F$ and the mass $m$. 
In $d\ge 2$ the geometric factor 
can also be computed explicitly \cite{LeschkeSobolevSpitzerLaplace} as
\begin{align}
\Phi(\Lambda,E_F,m)=\frac{2}{\Gamma\big(\frac{d+1}{2}\big)}\Big(\frac{p_{F}^2}{4\pi}\Big)^{\frac{d-1}{2}}|\partial\Lambda|,
\end{align}
where $\Gamma$ denotes Euler's gamma function.
\item For the R\'{e}nyi entropy functions $h_{\alpha}$, $\alpha\in\,]0,\infty[\,$, the coefficient $W(h,\Lambda,E_F,m)$ can also be computed explicitly as $W(h_{\alpha},\Lambda,E_F,m)=\frac{1+\alpha}{24\alpha}\Phi(\Lambda,E_F,m)$, see \cite{LeschkeSobolevSpitzerLaplace}. 
\item For positive mass $m>0$, $|E_F|>m$ and comparatively small Fermi momentum, i.e. $|E_F|-m\ll m$, the 
geometric factor $\Phi(\Lambda, E_{F},m)$ resembles the one in the corresponding asymptotic expansion for the Laplacian with non-relativistic Fermi momentum defined in terms of the kinetic energy and given by $p_{c}^{2}:=2m(|E_F|-m)$. More precisely, we have $\Phi(\Lambda,E_F,m)= J(\partial B_{p_{c}},\partial\Lambda)\big[1+O\big((|E_{F}|-m)/m\big)\big]$, where the function $J(\cdot,\cdot)$ is the corresponding geometric factor defined in \cite[Eq.\ (2)]{LeschkeSobolevSpitzerLaplace}.
\item 
		The statements of Theorem~\ref{Dirac} remain true in the case where the bigger volume 
		$\Lambda'$ is unbounded but has a bounded complement, see \cite[Remark 4.24]{BM-Widom}. In particular, one may set 
		$\Lambda'=\R^{d}$. This allows to relate the operators within the trace of 
		\eqref{Dirac-stat} to the operators studied in, e.g., 
		\cite{LeschkeSobSp.posT, LeschkeSobolevSpitzerMultiScale, Sobolev.2019.truncated, FinsterSob}, 
		where the volume term is subtracted directly. 
\item 		
	The factor $\frac{1}{n_{d}}$ on the left-hand side of \eqref{Dirac-stat} should be read as 	
	$\frac{1}{2} \frac{1}{n_{d}/2}$ with the first factor $\frac12$ accounting for the fact that the boundary of $\Lambda$ is counted twice by the trace. The second factor $\frac{1}{n_{d}/2}$ accounts for the number of spinor dimensions which contribute to the discontinuity in momentum space. As it turns out, each of these dimensions contributes the same amount to the enhanced area law. 
\end{enumerate}	
\end{rem}


\noindent
We now give a short overview over the strategy of the proof of Theorem \ref{Dirac}: We first evaluate 
\begin{align}
\tr_{L^2(\R^d)\otimes\C^{n_{d}}} 
			\Big[h\big(\mathbf{1}_{\Omega_{L}}\chi_{E_F}^{(b)}(D)
			 \mathbf{1}_{\Omega_{L}}\big) 
			\Big]
\end{align}
for a single bounded piece-wise $C^1$-admissible domain $\Omega$. By \eqref{Rewrite-PDO} and \eqref{Rewrite-L-PDO} this is equal to
\begin{align}\label{Goal-Single-Domain}
\tr_{L^2(\R^d)\otimes\C^{n_{d}}} 
			\Big[h\big(\mathbf{1}_{\Omega}\Op_{L}\big(\chi_{E_F}^{(b)}(\mathcal{D})\big)\mathbf{1}_{\Omega}\big)\Big].
\end{align}
In the cases where we establish an enhanced area law, we show that \eqref{Goal-Single-Domain} yields a volume term which is proportional to $|\Omega|$, an enhanced area term with coefficient $W(h,\Omega,E_F,m)$ and an error term of order $o(L^{d-1}\log L)$. The crucial part will be to study the symbol $\chi_{E_F}^{(b)}(\mathcal{D})$ and write it in such a way that the Widom--Sobolev formula for matrix-valued symbols \cite[Theorem 4.22, Remark 4.23]{BM-Widom} can be applied. 
In general, an application of \eqref{Dirac-Application-Function} with $a=\chi_{E_F}^{(b)}$ yields 
\begin{align}\label{Formula-Chi+,Chi-}
\chi_{E_F}^{(b)}\big(\mathcal{D}(\xi)\big)=\frac{1}{2}\Big(\big(\chi_{E_F}^{(b)}\big)_++\big(\chi_{E_F}^{(b)}\big)_-\Big)(\xi)\mathbb{1}_{n_d}+\frac{1}{2}\Big(\big(\chi_{E_F}^{(b)}\big)_+-\big(\chi_{E_F}^{(b)}\big)_-\Big)(\xi)\frac{\mathcal{D}(\xi)}{E(\xi)},
\end{align}
where $\big(\chi_{E_F}^{(b)}\big)_\pm$ are defined as in \eqref{a-pm-def}. The properties of these functions depend on the parameters $E_F$ and $m$ and will be investigated in the respective sections. After applying the Widom--Sobolev formula it remains to explicitly calculate the resulting coefficients. While doing so, we will frequently use the property that $h(0)=0$ for the function $h$ in Theorem \ref{Dirac} without explicitly stating it.

In the cases without an enhanced area law we estimate the trace norm 
\begin{align}
	\label{Goal-Single-Trace-Norm}
	\Big\|h\Big(\mathbf{1}_\Omega \Op_L\big(\chi_{E_F}^{(b)}(\mathcal{D})\big)\mathbf{1}_\Omega\Big)
	-\mathbf{1}_\Omega \Op_L\Big(h\big(\chi_{E_F}^{(b)}(\mathcal{D})\big)\Big)\mathbf{1}_\Omega\Big\|_1
\end{align}
and show that it only gives an area term. To do so, we first reduce the question to estimating a suitable Schatten-von Neumann norm by applying \cite[Lemma 4.19]{BM-Widom}. While  in the case $|E_F|\leq m\neq 0$ this estimate for the Schatten-von Neumann norm also follows from results in \cite{BM-Widom}, the case $E_F=m=0$ with $d\geq 2$ requires an additional estimate, as it turns out that the corresponding symbol is discontinuous at the origin. This estimate is provided in Lemma \ref{Schatten-Estimate}.

We start with the case $|E_F|>m$ in Section \ref{subsec:enhanced}, continue with the case $|E_F|\leq m\neq 0$ in Section \ref{subsec:area} and conclude with the case $E_F=m=0$ in Section \ref{subsec:point}. Afterwards we prove Theorem \ref{Dirac} in Section \ref{subsec:proof} by applying the results to $\Omega=\Lambda$, $\Omega=\Lambda'$ and $\Omega=\Lambda'\setminus\Lambda$ and by showing that the volume terms from \eqref{Goal-Single-Domain}, respectively the second terms in \eqref{Goal-Single-Trace-Norm}, add up to zero.

\subsection{The case \texorpdfstring{\boldmath$|E_F|>m$}{|E_F|>m}}
\label{subsec:enhanced}
It turns out that in this case the symbol $\chi_{E_F}^{(b)}(\mathcal{D})$ is discontinuous precisely at the boundary of the ball $B_{p_{F}}$.
We define $\Gamma:=\{\xi\in\R^d : \ E(\xi)<|E_F|\}=B_{p_{F}}$ as exactly this ball. This is clearly a bounded piece-wise $C^3$-domain and therefore satisfies the requirements of \cite[Theorem 4.22]{BM-Widom}. We now show the desired asymptotic formula for an arbitrary bounded piece-wise $C^1$-admissible domain $\Omega$. The next lemma treats the case $E_F<-m$. The case $E_F>m$ will be tackled in Lemma \ref{Lemma_E_F>m}.

\begin{lem}\label{Lemma_-E_F<-m}
Let $h$ be as in Theorem \ref{Dirac}, $\Omega$ be a bounded piece-wise $C^1$-admissible domain, $b\geq 0$, $E_F\in\R$ and $m\geq 0$ such that $E_F<-m$. Then, we have 
\begin{multline}
\frac{2}{n_{d}} \, \tr_{L^2(\R^d)\otimes\C^{n_{d}}}\Big[h\Big(\mathbf{1}_{\Omega}\Op_L\big(\chi_{E_F}^{(b)}(\mathcal{D})\big)\mathbf{1}_{\Omega}\Big)\Big] \\ = L^d \, V_{-}(h,b,E_{F},m) |\Omega |+L^{d-1}\log L \,  W(h,\Omega,E_F,m) +o(L^{d-1}\log L),
\end{multline}
as $L\rightarrow\infty$, where 
\begin{equation}
	V_{-}(h,b,E_{F},m) := \frac{1}{(2\pi)^d}\int_{\R^{d}} \big[h \circ \big(\chi_{E_F}^{(b)}\big)_- \big](\xi)\,\d\xi
\end{equation}
is independent of $n_{d}$ and $\Omega$.
\end{lem}

\begin{proof}
If $E_F<-m$, we have $\big(\chi_{E_F}^{(b)}\big)_+(\xi)=0$ and 
\begin{equation}
	\label{chi-minus}
	\big(\chi_{E_F}^{(b)}\big)_-(\xi)=1_{\Gamma^c}(\xi)\varphi\big(-E(\xi)+b\big)
	=:1_{\Gamma^c}(\xi)\tilde{\varphi}^{(b)}(\xi)
\end{equation} 
for the functions in \eqref{Formula-Chi+,Chi-}, which allows to rewrite  
\begin{align}
	\label{lemma3.3-start}
	h\Big(\mathbf{1}_{\Omega}\Op_L\big(\chi_{E_F}^{(b)}(\mathcal{D})\big)\mathbf{1}_{\Omega}\Big) 
	= h\Big(\mathbf{1}_{\Omega}\Op_L(\mathbf{1}_{\Gamma^c}) \Op_L\big(\tfrac{1}{2}\tilde{\varphi}^{(b)}
		(\mathbb{1}_{n_d}-\tfrac{\mathcal{D}}{E})\big)\Op_L(\mathbf{1}_{\Gamma^c})\mathbf{1}_{\Omega}\Big).
\end{align}
Recall that the matrix-valued symbol $\frac{\mathcal{D}}{E}$ is smooth, except at the origin in the case $m=0$. 
Due to the presence of the projection $\Op_L(\mathbf{1}_{\Gamma^c})$, we can modify the function $\tilde{\varphi}^{(b)}$ inside of $\Gamma$ without changing \eqref{lemma3.3-start}. In particular, we can replace $\tilde{\varphi}^{(b)}$ by a smooth and compactly supported function $\varphi^{(b)}\in C_c^\infty(\R^d)$ with $\varphi^{(b)}(0)=0$, as $\overline{\Gamma^c}\cap \{0\}=\emptyset$. Then, the symbol $\frac{1}{2}\varphi^{(b)}( \mathbb{1}_{n_d}-\frac{\mathcal{D}}{E})$ is smooth, bounded and compactly supported. As $\Omega$ is a bounded piece-wise $C^1$-admissible domain, $\Gamma$ is a bounded piece-wise $C^3$-admissible domain and $h$ satisfies Assumption \ref{H-Functions} with $h(0)=0$, all the requirements of \cite[Theorem 4.22]{BM-Widom} are fulfilled, and an application of this theorem with $A_1:=0, A_2:=\frac{1}{2}\varphi^{(b)}( \mathbb{1}_{n_d}-\frac{\mathcal{D}}{E})$ yields the asymptotic expansion
\begin{align}\label{Asymptotics-Last-Paper}
	\tr_{L^2(\R^d)\otimes\C^n} \Big[h\Big(\mathbf{1}_{\Omega}\Op_L(\mathbf{1}_{\Gamma^c})
		&\Op_L\big(\tfrac{1}{2}\varphi^{(b)}( \mathbb{1}_{n_d}-\tfrac{\mathcal{D}}{E})\big)
		\Op_L(\mathbf{1}_{\Gamma^c})\mathbf{1}_{\Omega}\Big)\Big] \nonumber \\ 
	& = L^d \, \mathfrak{W}_0\Big(\tr_{\C^n}\!\big[h\big(\tfrac{1}{2}\varphi^{(b)}
			( \mathbb{1}_{n_d}-\tfrac{\mathcal{D}}{E} ) \big)\big];\Omega,\Gamma^c\Big)\nonumber\\
	& \quad + L^{d-1}\log L \ \mathfrak{W}_1\Big(\mathfrak{U}\big(h;0,\tfrac{1}{2}\varphi^{(b)}
			( \mathbb{1}_{n_d}-\tfrac{\mathcal{D}}{E} )\big);\partial\Omega,\partial\Gamma\Big)\nonumber\\ 
	& \quad + o(L^{d-1}\log L),
\end{align}
as $L\rightarrow\infty$. Here, the coefficients are defined as follows for a bounded scalar-valued symbol $a$ and admissible domains $\Omega,\Xi$ with $1_{\Omega}$ and $1_{\Xi}a$ having compact support
\begin{align}
\label{coeff-w0-def}
	\mathfrak{W}_0(a;\Omega,\Xi):= \frac{ |\Omega|}{(2\pi)^d}\int_{\Xi} a(\xi) \,\d\xi \
\end{align}
and 
\begin{align}
	\label{coeff-w1-def}
	\mathfrak{W}_1(a;\partial\Omega,\partial\Xi):= 
	\LEFTRIGHT\{.{\begin{array}{l@{\quad}l} \displaystyle |\partial\Omega|\sum_{\xi\in\partial\Xi}a(\xi), 
			& \text{for~} d=1, \\[4ex]
		\displaystyle\frac{1}{(2\pi)^{d-1}}\int_{\partial\Omega}\int_{\partial\Xi} a(\xi)\,
			\vert n_{\partial \Omega}(x)\cdot n_{\partial \Xi}(\xi)\vert \,\d S(\xi)\, \d S(x), 
			& \text{for~} d\ge 2. \end{array}}
\end{align}
We note that if the symbol $a$ is equal to some constant $a_0$ on $\partial\Xi$ and $\partial\Xi = \partial B_{p_{F}}$, 
then 
\begin{equation}
	\label{frakW-Phi}
	\mathfrak{W}_1(a;\partial\Omega,\partial\Gamma) = a_0\Phi(\Omega,E_F,m),
\end{equation} 
where $\Phi(\Omega,E_F,m)$ is defined in Theorem \ref{Dirac}.
The symbol $\mathfrak{U}(g;A_1,A_2)$ appearing in the first argument of the coefficient $\mathfrak{W}_1$ in \eqref{Asymptotics-Last-Paper} is given by
\begin{equation}
	\label{frakU-def}
	\mathfrak{U}(g;A_1,A_2):=  \frac{1}{(2\pi)^2}\int_0^1 
	\frac{\tr_{\C^{n_{d}}}\! \big[ g\big(A_1t+A_2(1-t)\big)-g(A_1)t-g(A_2)(1-t) \big]}{t(1-t)}\;\d t
\end{equation}
and depends on a H\"older continuous function $g: \R \to\C$ and bounded matrix-valued symbols $A_{1}, A_{2}$. 

In order to compute the matrix traces in the coefficients, it will be convenient to diagonalise the symbol 
$\frac{1}{2}\varphi^{(b)}( \mathbb{1}_{n_d}-\tfrac{\mathcal{D}}{E})
	=\frac{1}{2}U^{-1}\varphi^{(b)}(\mathbb{1}_{n_d}-\beta) U$ argument-wise as in \eqref{Dirac-Application-Function}.
This gives for any $z\in\C$
\begin{align}
	\label{E_F<-m-Calc1}
	\tr_{\C^{n_{d}}}\! \Big[h\big(\tfrac{z}{2}\varphi^{(b)}(\mathbb{1}_{n_d}- \tfrac{\mathcal{D}}{E}) \big)\Big] 
	& = \tr_{\C^{n_{d}}}\! \Big[U^{-1} h\big(\tfrac{z}{2}\varphi^{(b)}(\mathbb{1}_{n_d}-\beta) \big)U\Big]
		\nonumber \\
	&	= \tr_{\C^{n_{d}}}\! \begin{pmatrix} 0 & 0 \\ 0 & h(z\varphi^{(b)})\mathbb{1}_{\frac{n_d}{2}} 
			  							 \end{pmatrix}  
		= \frac{n_d}{2} \, h(z\varphi^{(b)})
\end{align}
and implies for $z=1-t$ and $z=1$, respectively,
\begin{align}
	\label{Lemma3.3-frakU}
	\mathfrak{U}\big( & h;0,\tfrac{1}{2}\varphi^{(b)}(\mathbb{1}_{n_d}- \tfrac{\mathcal{D}}{E}) \big) \nonumber \\ 
	& = \frac{1}{(2\pi)^2}\int_0^1 \frac{\tr_{\C^{n_{d}}} \big[h\big(\frac{1}{2} \varphi^{(b)}
				(\mathbb{1}_{n_d}- \tfrac{\mathcal{D}}{E}) (1-t)\big) 
				- h\big(\frac{1}{2}\varphi^{(b)}(\mathbb{1}_{n_d}-\tfrac{\mathcal{D}}{E})\big)(1-t)\big]}{t(1-t)} \;\d t
		\nonumber \\ 
	& = \frac{n_{d}}{2}\frac{1}{(2\pi)^2}\int_0^1 \frac{h\big(\varphi^{(b)}(1-t)\big) - 
				h(\varphi^{(b)})(1-t)}{t(1-t)}
			\,\d t
		= \frac{n_d}{2}\;\mathfrak{A}(h;\varphi^{(b)}),
\end{align}
where we introduced
\begin{equation}
	\label{frakA-def}
	\mathfrak{A}(h;a):=\frac{1}{(2\pi)^2}\int_0^1 \frac{h(at)-h(a)t}{t(1-t)}\;\d t 
\end{equation}	
for a bounded scalar-valued symbol $a$.

We conclude from \eqref{E_F<-m-Calc1} with $z=1$ that 
\begin{multline}
\label{Lemma3.3-frakW0}
\mathfrak{W}_0\Big(\tr_{\C^{n_{d}}}\! \big[h\big(\tfrac{1}{2}\varphi^{(b)}( \mathbb{1}_{n_d}-\tfrac{\mathcal{D}}{E})\big)\big];\Omega,\Gamma^c\Big) \\
=\frac{n_d}{2} \, \mathfrak{W}_0\big(h(\varphi^{(b)});\Omega,\Gamma^c\big)  = \frac{n_d}{2}\frac{|\Omega|}{(2\pi)^d}\int_{\Gamma^c} \big(h(\varphi^{(b)})\big)(\xi)\;\d\xi
= \frac{n_d}{2}\, |\Omega| V_{-}(h,b,E_{F},m),
\end{multline}
where the last equality rests on \eqref{chi-minus}.

Now we return to \eqref{Lemma3.3-frakU}. Since $\varphi^{(b)}|_{\partial\Gamma}=1$ we see that 
$\mathfrak{A}(h;\varphi^{(b)})|_{\partial\Gamma} = \frac{1}{(2\pi)^2}\int_0^1 \frac{h(t)-h(1)t}{t(1-t)}\;\d t$
is also constant and, thus we infer from \eqref{frakW-Phi} that 
\begin{multline}\label{W_1_Enhanced}
\mathfrak{W}_1\Big(\mathfrak{U}\big(h;0,\tfrac{1}{2}\varphi^{(b)}( \mathbb{1}_{n_d}-\tfrac{\mathcal{D}}{E}) \big);\partial\Omega,\partial\Gamma\Big)=\frac{n_d}{2}\; \mathfrak{W}_1\big(\mathfrak{A}(h;\varphi^{(b)});\partial\Omega,\partial\Gamma\big) \\ =\frac{n_d}{2}\frac{1}{(2\pi)^2}\int_0^1\frac{h(t)-h(1)t}{t(1-t)}\; \d t  \ \Phi(\Omega,E_F,m)=\frac{n_d}{2}\,W(h,\Omega,E_F,m).
\end{multline}
Therefore, the claim follows with \eqref{Lemma3.3-frakW0} and \eqref{Asymptotics-Last-Paper}.
\end{proof}

\begin{lem}\label{Lemma_E_F>m}
Let $h$ be as in Theorem \ref{Dirac}, $\Omega$ be a bounded piece-wise $C^1$-admissible domain, $b\geq 0$, $E_F\in\R$ and $m\geq 0$ such that $E_F>m$. Then, we have 
\begin{multline}
	\frac{2}{n_{d}}\, \tr_{L^2(\R^d)\otimes\C^{n_{d}}}\Big[h\Big(\mathbf{1}_{\Omega}\Op_L\big(\chi_{E_F}^{(b)}(\mathcal{D})\big)\mathbf{1}_{\Omega}\Big)\Big] \\ = L^d V_{+}(h,b,E_{F},m)|\Omega| + L^{d-1}\log L \; W(h,\Omega,E_F,m) +o(L^{d-1}\log L),
\end{multline}
as $L\rightarrow\infty$, where  
\begin{equation}
	V_{+}(h,b,E_{F},m) := V_{-}(h,b,E_{F},m) + \frac{2|B_{p_{F}}|}{(2\pi)^{d}}\, h(1)
\end{equation}
is independent of $n_{d}$ and $\Omega$.
\end{lem}

\begin{proof}
In the case $E_F>m$ we obtain $\big(\chi_{E_F}^{(b)}\big)_+(\xi)=1_{\Gamma}(\xi)$ and $\big(\chi_{E_F}^{(b)}\big)_-(\xi)=\varphi(-E(\xi)+b)=\tilde{\varphi}^{(b)}(\xi)$ for the functions from \eqref{Formula-Chi+,Chi-}. Further we have $\big(\chi_{E_F}^{(b)}\big)_-|_\Gamma=1$ by the definition of $\tilde{\varphi}^{(b)}$. Therefore, we can write $\big(\chi_{E_F}^{(b)}\big)_- (\xi)=1_{\Gamma}(\xi) + 1_{\Gamma^c}(\xi) \tilde{\varphi}^{(b)}(\xi)$. Applying \eqref{Formula-Chi+,Chi-}, 
the operator in question reads
\begin{multline}
\label{lemma3.4-start}
h\big(\mathbf{1}_{\Omega}\Op_L(\chi_{E_F}^{(b)}(\mathcal{D}))\mathbf{1}_{\Omega}\big) \\ = h\Big(\mathbf{1}_{\Omega}\big\{\Op_L(\mathbf{1}_\Gamma)+\Op_L(\mathbf{1}_{\Gamma^c})\Op_L\big(\tfrac{1}{2}\tilde{\varphi}^{(b)}( \mathbb{1}_{n_d}-\tfrac{\mathcal{D}}{E})\big)\Op_L(\mathbf{1}_{\Gamma^c})\big\}\mathbf{1}_{\Omega}\Big).
\end{multline}
As in the proof of Lemma \ref{Lemma_-E_F<-m} after \eqref{lemma3.3-start}, we can replace the function $\tilde{\varphi}^{(b)}$ in \eqref{lemma3.4-start} by a function $\varphi^{(b)}\in C_c^\infty(\R^d)$ with $\varphi^{(b)}(0)=0$ without changing the result. 
This guarantees that the symbols $A_1:=\mathbb{1}_{n_d}$ and $A_2:=\frac{1}{2}\varphi^{(b)}( \mathbb{1}_{n_d}-\frac{\mathcal{D}}{E})$ are smooth, bounded and the latter also compactly supported. Thus, we can apply \cite[Theorem 4.22]{BM-Widom}  with the symbols $A_1$ and $A_2$ as above.
This yields the asymptotic expansion
\begin{align}
	\label{lemma3.4-next}
	\tr_{L^2(\R^d)\otimes\C^{n_{d}}} & \Big[h\Big(\mathbf{1}_{\Omega} \big\{\Op_L(\mathbf{1}_\Gamma)
		+ \Op_L(\mathbf{1}_{\Gamma^c})\Op_L\big(\tfrac{1}{2}\varphi^{(b)}( \mathbb{1}_{n_d}-\tfrac{\mathcal{D}}{E})\big)\Op_L(\mathbf{1}_{\Gamma^c})\big\}\mathbf{1}_{\Omega}\Big)\Big] 
		\nonumber \\ 
	& = L^d \,\mathfrak{W}_0\Big(\tr_{\C^{n_{d}}}\!\big[h(\mathbb{1}_{n_d})\big];\Omega,\Gamma\Big)
		\nonumber\\
	& \quad + L^d \,\mathfrak{W}_0\Big(\tr_{\C^{n_{d}}}\! \big[h\big(\tfrac{1}{2}\varphi^{(b)}
			( \mathbb{1}_{n_d}-\tfrac{\mathcal{D}}{E})\big)\big];\Omega,\Gamma^c\Big)\nonumber\\
	& \quad + L^{d-1}\log L \; \mathfrak{W}_1\Big(\mathfrak{U}\big(h;\mathbb{1}_{n_d},\tfrac{1}{2}\varphi^{(b)}( \mathbb{1}_{n_d}-\tfrac{\mathcal{D}}{E})\big);\partial\Omega,\partial\Gamma\Big)\nonumber\\ 
	& \quad +o(L^{d-1}\log L),
\end{align}
as $L\rightarrow\infty$.

We observe 
\begin{align}
\label{trace-one}
\tr_{\C^{n_{d}}}[h(\mathbb{1}_{n_d})]=n_{d}h(1)
\end{align}
and claim 
\begin{multline}
	\label{frakW0-3.4}
	\mathfrak{W}_0\Big(\tr_{\C^{n_{d}}} \! \big[h(\mathbb{1}_{n_d})\big];\Omega,\Gamma\Big)
		+ \mathfrak{W}_0  \Big(\tr_{\C^{n_{d}}} \! \big[ h\big(\tfrac{1}{2}\varphi^{(b)}
			(\mathbb{1}_{n_d}-\tfrac{\mathcal{D}}{E})\big)\big];\Omega,\Gamma^c\Big)  \\ 
	 = \frac{n_{d}}{2} \, |\Omega| \,\bigg[ \frac{2|\Gamma|}{(2\pi)^d} \; h(1)
				+ V_{-}(h,b,E_{F},m) \bigg].
\end{multline}
Here, the first coefficient was evaluated with the definition \eqref{coeff-w0-def} and the second one with the identity \eqref{Lemma3.3-frakW0}.

In order to treat the coefficient of the enhanced area term we observe that by argument-wise diagonalisation 
$\tfrac{1}{2}\varphi^{(b)}( \mathbb{1}_{n_d}-\tfrac{\mathcal{D}}{E}) =\tfrac{1}{2}U^{-1}\varphi^{(b)}(\mathbb{1}_{n_d}-\beta) U$ as in \eqref{Dirac-Application-Function} we get for 
$\xi\in\partial\Gamma$, for which  $\varphi^{(b)}(\xi)=1$, 
\begin{align}
	\tr_{\C^{n_{d}}} \! \Big[h \Big(\mathbb{1}_{n_d}t +  & \tfrac{1}{2}\varphi^{(b)}(\xi) \big( \mathbb{1}_{n_d}-\tfrac{\mathcal{D}}{E}(\xi)\big) (1-t)\Big) \Big] \notag\\
	& = \tr_{\C^{n_{d}}} \! \Big[ U^{-1}(\xi) h\Big(\mathbb{1}_{n_d}t + 
			\tfrac{1}{2}(\mathbb{1}_{n_d}-\beta)(1-t)\Big)     U(\xi) \Big] \notag\\
	& = \tr_{\C^{n_{d}}}\! \begin{pmatrix}  h(t) \mathbb{1}_{\frac{n_d}{2}}  & 0 \\ 0 & h(1)\mathbb{1}_{\frac{n_d}{2}} 
			  							 \end{pmatrix}  	
		= \frac{n_{d}}{2}\, [h(t) + h(1)]	.							 	
\end{align} 
Using this identity, \eqref{trace-one} and \eqref{E_F<-m-Calc1} with $z=1$ and $\varphi^{(b)}=1$, we obtain from the definition \eqref{frakU-def} 
\begin{align}\label{U-E_F>m}
 \mathfrak{U}\big(h;\mathbb{1}_{n_d} &,\tfrac{1}{2}\varphi^{(b)}( \mathbb{1}_{n_d}-\tfrac{\mathcal{D}}{E})\big) (\xi)  \nonumber \\ 
	& = \frac{1}{(2\pi)^2}\int_0^1 \frac{1}{t(1-t)} \, \tr_{\C^{n_{d}}}\bigg[  h\Big(\mathbb{1}_{n_d}t + 
			\tfrac{1}{2}\varphi^{(b)}(\xi)\big(\mathbb{1}_{n_d}- \tfrac{\mathcal{D}}{E}(\xi) \big)(1-t)\Big) 		\nonumber \\ 
	& \hspace{4.8cm}  -   h(\mathbb{1}_{n_d})t 	- h\Big( \tfrac{1}{2}\varphi^{(b)}(\xi)\big(\mathbb{1}_{n_d} 
			- \tfrac{\mathcal{D}}{E}(\xi) \big)\Big) (1-t) \bigg] \, \d t 
		\nonumber \\ 
	& = \frac{n_{d}}{2} \, \frac{1}{(2\pi)^2} \int_0^1 \frac{h(t)+h(1)
				- 2h(1)t - h(1)(1-t)}{t(1-t)} \, \d t \nonumber \\
	& =  \frac{n_{d}}{2}\; \mathfrak{A}(h;1), 
\end{align}
where $\mathfrak{A}$ was introduced in \eqref{frakA-def}. Accordingly, the symbol \eqref{U-E_F>m} is constant on 
$\partial\Gamma$, and the relation \eqref{frakW-Phi} yields 
\begin{align}
	\label{frakW1-3.4}
	\mathfrak{W}_1\Big(\mathfrak{U}\big(h;\mathbb{1}_{n_d},\tfrac{1}{2}\varphi^{(b)}
		( \mathbb{1}_{n_d}-\tfrac{\mathcal{D}}{E})\big);\partial\Omega,\partial\Gamma\Big)
	& = \frac{n_{d}}{2}\frac{1}{(2\pi)^2}\int_0^1\frac{h(t)-h(1)}{t(1-t)}\,\d t  \;
			\Phi(\Omega,E_F,m) \nonumber\\ 
	&	= \frac{n_{d}}{2}\, W(h,\Omega,E_F,m).
\end{align}
Therefore, the claim follows from \eqref{frakW1-3.4}, \eqref{frakW0-3.4}, \eqref{lemma3.4-next} and 
\eqref{lemma3.4-start}.
\end{proof}


\subsection{The case \texorpdfstring{\boldmath$|E_F|\leq m\neq 0$}{|E_F| <= m  positive}}
\label{subsec:area}

The key observation in this case is that the discontinuity of the function $\chi_{E_F}^{(b)}$ does 
not affect the symbol $\chi_{E_F}^{(b)}\big(\mathcal{D}\big)$. This is because there exists $\tilde\chi_{E_F}^{(b)} \in C_{c}^{\infty}(\R)$ with $\supp\big(\chi_{E_F}^{(b)}  - \tilde\chi_{E_F}^{(b)} \big) \subseteq [-m,m]$. Therefore,
Remark~\ref{Remark-Application-Function}(d) allows to replace $\chi_{E_F}^{(b)}$ by $\tilde\chi_{E_F}^{(b)}$ without changing the symbol, $\tilde\chi_{E_F}^{(b)}(\mathcal{D}) = \chi_{E_F}^{(b)}(\mathcal{D})$ Lebesgue-a.e., and  Remark~\ref{Remark-Application-Function}(e) guarantees smoothness of the symbol $\tilde\chi_{E_F}^{(b)}(\mathcal{D})$. We now show that this is sufficient to obtain the desired estimate for \eqref{Goal-Single-Trace-Norm} from \cite[Lemmas 3.3 and 4.19]{BM-Widom}.

\begin{lem}\label{Lemma_E_F=<m>0}
Let $h$ be as in Theorem \ref{Dirac}, $\Omega$ be a bounded admissible domain, $b\geq 0$, $E_F\in\R$ and $m>0$ such that $|E_F|\leq m$. Then for every $L\geq 1$, we have
\begin{align}
\Big\| h\Big(\mathbf{1}_\Omega \Op_L\big(\chi_{E_F}^{(b)}(\mathcal{D})\big)\mathbf{1}_\Omega\Big)-\mathbf{1}_\Omega \Op_L\Big(h\big(\chi_{E_F}^{(b)}(\mathcal{D})\big)\Big)\mathbf{1}_\Omega\Big\|_1 
\leq CL^{d-1},
\end{align}
where the constant $C>0$ is independent of $L$.
\end{lem}
\begin{proof}
As $\chi_{E_F}^{(b)}$ and thus the operator $\Op_L\big(\chi_{E_F}^{(b)}(\mathcal{D})\big)$ are bounded, we can assume that the function $h$ is compactly supported. An application of \cite[Lemma 4.19]{BM-Widom} with $A=\Op_L\big(\chi_{E_F}^{(b)}(\mathcal{D})\big)$, $P=\mathbf{1}_\Omega$ and arbitrary 
$q\in \,]0,\gamma[\,$ yields
\begin{align}
\Big\| h\Big(\mathbf{1}_\Omega \Op_L\big(\chi_{E_F}^{(b)}(\mathcal{D})\big)\mathbf{1}_\Omega\Big)-\mathbf{1}_\Omega \Op_L\Big(h\big(\chi_{E_F}^{(b)}(\mathcal{D})\big)\Big)\mathbf{1}_\Omega\Big\|_1
\leq C_1 \big\|\mathbf{1}_\Omega \Op_L \big(\chi_{E_F}^{(b)}(\mathcal{D})\big)\mathbf{1}_{\Omega^c} \big\|_q^q,
\end{align}
where the constant $C_1>0$ does not depend on the parameter $L$. Therefore, it suffices to show
\begin{align}
\big\|\mathbf{1}_\Omega \Op_L \big(\chi_{E_F}^{(b)}(\mathcal{D})\big)\mathbf{1}_{\Omega^c} \big\|_q^q \leq C_2L^{d-1}
\end{align}
for some constant $C_2>0$ which does not depend on $L$.
As discussed at the beginning of this section, we can replace the discontinuous symbol $\chi_{E_F}^{(b)}(\mathcal{D})$ by $\tilde\chi_{E_F}^{(b)}(\mathcal{D}) \in C_{c}^{\infty}(\R^{d})$. As $\Omega$ is bounded, there exists a function $\phi\in C_c^\infty(\R^d)$ with $\phi|_{\Omega}=1$, and we have $\phi \tilde\chi_{E_F}^{(b)}(\mathcal{D}) \in C_{c}^{\infty}(\R^{d} \times \R^{d})$. As $\Omega$ is also admissible, we can apply \cite[Lemma~3.3]{BM-Widom} with 
$A=\phi \tilde\chi_{E_F}^{(b)}(\mathcal{D})$ and $\Lambda=\Omega$ to get the desired result.
\end{proof}


\subsection{The case \texorpdfstring{\boldmath$E_F=m=0$}{E_F=m=0}}
\label{subsec:point}

If $|E_F|=m=0$, we obtain  the functions $\big(\chi_0^{(b)}\big)_+=0$ and $\big(\chi_0^{(b)}\big)_-=1_{\{y\in\R\,:\, y<0\}}(-E)\varphi(-E+b)=1_{\R^{d}\setminus\{0\}}\varphi(-E+b)$. As in the proof of Lemma \ref{Lemma_-E_F<-m} we define the function $\tilde{\varphi}^{(b)}$ by $\tilde{\varphi}^{(b)}(\xi):=\varphi(-E(\xi)+b)$ for every $\xi\in\R^d$. It only differs from $\big(\chi_0^{(b)}\big)_-$ at the origin $\xi=0$, i.e.\ on a set of Lebesgue measure zero. Therefore, after an application of \eqref{Formula-Chi+,Chi-}, the relevant operator reads
\begin{align}\label{Operator-m=0}
\mathbf{1}_\Omega \Op_L\big(\chi_0^{(b)}(\mathcal{D})\big)\mathbf{1}_\Omega=\mathbf{1}_\Omega \Op_L\Big(\tfrac{1}{2} \tilde{\varphi}^{(b)}\big(\mathbb{1}_{n_d}-\tfrac{\mathcal{D}}{E}\big)\Big)\mathbf{1}_\Omega.
\end{align}
We observe that the symbol $\frac{1}{2} \tilde{\varphi}^{(b)}\big(\mathbb{1}_{n_d}-\frac{\mathcal{D}}{E}\big) \in C_{c}^{\infty}(\R^{d}\setminus \{0\})$ has a discontinuity at the origin. It will turn out that this zero-dimensional discontinuity leads to the same behaviour as in Section \ref{subsec:enhanced} in dimension $1$. In higher dimensions, it will lead to the same behaviour as in Section \ref{subsec:area}.

We will first consider the case $d=1$, where, accordingly $n_1=2$.
\begin{lem}\label{Lemma_E_F=m=0_d=1}
Let $h$ be as in Theorem \ref{Dirac}, $\Omega$ be a bounded piece-wise $C^1$-admissible domain, $b\geq 0$, $E_F=m=0$ and $d=1$. Then, we have
\begin{multline}
\frac{2}{n_{1}}\, \tr_{L^2(\R)\otimes\C^2}\Big[h\Big(\mathbf{1}_{\Omega}\Op_L\big(\chi_{0}^{(b)}(\mathcal{D})\big)\mathbf{1}_{\Omega}\Big)\Big] 
= L V_{0}(h,b) |\Omega | + \log L \, W(h,\Omega,0,0) +o(\log L),
\end{multline}
as $L\rightarrow\infty$, where 
\begin{equation}
	V_{0}(h,b) := \frac{1}{2\pi} \int_{\R} \big[h \comp \big(\chi_{0}^{(b)}\big)_{-}\big](\xi)\,\d\xi 
\end{equation}
is independent of $\Omega$.
\end{lem}
\begin{proof}

If $d=1$ and $m=0$, there is only one relevant Dirac matrix. We choose 
\begin{align}
\mathcal{D}(\xi)=\sum_{k=1}^1 \alpha_k \xi_k +m\beta = \alpha_1 \xi=\begin{pmatrix}
0 & \xi \\
 \xi & 0
\end{pmatrix}.
\end{align}
We define the (unbounded) admissible domain $\Gamma := \, ]-\infty,0[\,$ and write
\begin{align}
\mathbb{1}_{2}-\frac{\mathcal{D}}{E}=\begin{pmatrix}
1 & -\sgn \\
 -\sgn & 1
\end{pmatrix}=\begin{pmatrix}
1 & 1 \\
 1 & 1
\end{pmatrix}\mathbf{1}_\Gamma+\begin{pmatrix}
1 & -1 \\
 -1 & 1
\end{pmatrix}\mathbf{1}_{\Gamma^c}
= (\mathbb{1}_{2}+\alpha_1) \mathbf{1}_\Gamma + (\mathbb{1}_{2}-\alpha_1) \mathbf{1}_{\Gamma^c},
\end{align}
where the second equality holds for all $\xi\in\R\setminus\{0\}$, as we use the conventions 
$\frac{\mathcal{D}}{E}(0)=0$ and $\sgn(0)=0$ for the sign function $\sgn$ on $\R$. With this, the operator on the right-hand side of \eqref{Operator-m=0} reads
\begin{equation}
\label{operator-00}
\mathbf{1}_\Omega\Big[\Op_L\Big(\tfrac{1}{2}\tilde{\varphi}^{(b)}(\mathbb{1}_{2}+\alpha_1)\mathbf{1}_\Gamma\Big)+\Op_L\Big(\tfrac{1}{2}\tilde{\varphi}^{(b)}(\mathbb{1}_{2}-\alpha_1)\mathbf{1}_{\Gamma^c}\Big)\Big]\mathbf{1}_\Omega.
\end{equation}
We note that the symbol $\tilde{\varphi}^{(b)}$ is smooth, bounded and compactly supported. Therefore, the symbols $A_1:=\tfrac{1}{2}\tilde{\varphi}^{(b)}(\mathbb{1}_{2}+\alpha_1)$ and $A_2:=\tfrac{1}{2}\tilde{\varphi}^{(b)}(\mathbb{1}_{2}-\alpha_1)$ are also smooth, bounded and compactly supported. An application of \cite[Theorem 4.22, Remark 4.23]{BM-Widom} yields the asymptotic expansion
\begin{align}
	\label{asymptotics-00}
	\tr_{L^2(\R)\otimes\C^2}\Big[h\Big(\mathbf{1}_\Omega\Op_L\Big(\tfrac{1}{2} & 
			\tilde{\varphi}^{(b)}\big((\mathbb{1}_{2}+\alpha_1)\mathbf{1}_\Gamma
				+(\mathbb{1}_{2}-\alpha_1)\mathbf{1}_{\Gamma^c}\big)\Big)\mathbf{1}_\Omega\Big)\Big] \nonumber \\ 
	& = L \;\mathfrak{W}_0\Big(\tr_{\C^2}\big[h\big(\tfrac{1}{2}\tilde{\varphi}^{(b)}( \mathbb{1}_{2}+\alpha_1)\big)\big];
			\Omega,\Gamma\Big) \nonumber\\
	&\quad + L \; \mathfrak{W}_0\Big(\tr_{\C^2}\big[h\big(\tfrac{1}{2}\tilde{\varphi}^{(b)}( \mathbb{1}_{2}-\alpha_1)\big)\big];
			\Omega,\Gamma^c\Big) \nonumber\\
	&\quad + \log L \; \mathfrak{W}_1\Big(\mathfrak{U}\big(h;\tfrac{1}{2}\tilde{\varphi}^{(b)}( \mathbb{1}_{2} +
			\alpha_1),\tfrac{1}{2}\tilde{\varphi}^{(b)}( \mathbb{1}_{2}-\alpha_1)\big);\partial\Omega,\partial\Gamma\Big)\nonumber\\ 
	&\quad +o(\log L),
\end{align}
as $L\rightarrow\infty$. 
We infer from the eigenvalues of the matrices $\mathbb{1}_{2}+\alpha_1$, $\mathbb{1}_{2}-\alpha_1$ and $\mathbb{1}_{2}+(2t-1)\alpha_1$ that
\begin{align}
\tr_{\C^2}\big[h\big(\tfrac{1}{2}\tilde{\varphi}^{(b)}(\mathbb{1}_{2}+\alpha_1) \big)\big]=h(\tilde{\varphi}^{(b)})=\tr_{\C^2}\big[h\big(\tfrac{1}{2}\tilde{\varphi}^{(b)}(\mathbb{1}_{2}-\alpha_1) \big)\big]
\end{align}
and
\begin{equation}
	\tr_{\C^2}\big[h\big(\tfrac{1}{2}\tilde{\varphi}^{(b)}(\mathbb{1}_{2}+\alpha_1)t +
		\tfrac{1}{2}\tilde{\varphi}^{(b)}(\mathbb{1}_{2}-\alpha_1)(1-t)\big)\big]
	= h\big(\tilde{\varphi}^{(b)}t\big)+h\big(\tilde{\varphi}^{(b)}(1-t)\big)	
\end{equation}
so that 
\begin{align}\label{E_F=m=0_U-Rechnung}
	\mathfrak{U} & \Big(h;\tfrac{1}{2}\tilde{\varphi}^{(b)}( \mathbb{1}_{2}+\alpha_1),\tfrac{1}{2}\tilde{\varphi}^{(b)}
		( \mathbb{1}_{2}-\alpha_1)\Big) \nonumber \\
	& = \frac{1}{(2\pi)^2} \int_0^1 \frac{h\big(\tilde{\varphi}^{(b)}t\big)+h\big(\tilde{\varphi}^{(b)}(1-t)\big)
			- h\big(\tilde{\varphi}^{(b)}\big)t-h\big(\tilde{\varphi}^{(b)}\big)(1-t)}{t(1-t)}\;\d t \nonumber\\ 
	& = 2\; \mathfrak{A}\big(h;\tilde{\varphi}^{(b)}\big).
\end{align}
Thus, we obtain the volume coefficient
\begin{multline}
	\label{W0-00}
	\mathfrak{W}_0\Big(\tr_{\C^2}\big[h\big(\tfrac{1}{2}\tilde{\varphi}^{(b)}( \mathbb{1}_{2}+\alpha_1)\big)\big];\Omega,\Gamma\Big)
	+ \mathfrak{W}_0\Big(\tr_{\C^2}\big[h\big(\tfrac{1}{2}\tilde{\varphi}^{(b)}(\mathbb{1}_{2}-\alpha_1)\big)\big];\Omega,
			\Gamma^c\Big) \\ 
	= \frac{|\Omega|}{2\pi} \int_{\R} \big(h(\tilde{\varphi}^{(b)})\big)(\xi)\;\d\xi 
	= |\Omega| V_{0}(h,b)
\end{multline}
and the coefficient of the enhanced area term 
\begin{multline}
	\label{W1-00}
	\mathfrak{W}_1\Big(\mathfrak{U}\big(h;\tfrac{1}{2}\tilde{\varphi}^{(b)}( \mathbb{1}_{2}+\alpha_1),\tfrac{1}{2}
		\tilde{\varphi}^{(b)}( \mathbb{1}_{2}-\alpha_1)\big);\partial\Omega,\partial\Gamma\Big)
	= 2\; \mathfrak{W}_1\big(\mathfrak{A}\big(h;\tilde{\varphi}^{(b)}\big);\partial\Omega,\partial\Gamma\big) \\ =\frac{\Phi(\Omega,0,0)}{(2\pi)^2}\int_0^1\frac{h(t)-h(1)t}{t(1-t)}\d t  \ =W(h,\Omega,0,0),
\end{multline}
where we used \eqref{frakW-Phi} together with $\tilde{\varphi}^{(b)}|_{\partial\Gamma}=\tilde{\varphi}^{(b)}(0)=1$
for the second equality.
Now, the lemma follows from \eqref{W1-00}, \eqref{W0-00}, \eqref{asymptotics-00} and \eqref{operator-00}.
\end{proof}

\begin{rem}
Upon comparing the proof of Lemma~\ref{Lemma_E_F=m=0_d=1} with the proofs of Lemma~\ref{Lemma_-E_F<-m} and~\ref{Lemma_E_F>m} in $d=1$ and recalling $n_{1}=2$, we observe an additional factor of 2 in \eqref{E_F=m=0_U-Rechnung} as compared to 
\eqref{Lemma3.3-frakU} or \eqref{U-E_F>m}. On the other hand, we have $|\partial\Gamma|=1$ in the proof of Lemma~\ref{Lemma_E_F=m=0_d=1}, whereas $|\partial\Gamma|=2$ in the proofs of Lemma~\ref{Lemma_-E_F<-m} and~\ref{Lemma_E_F>m}.
Still, this allows to write the coefficient of the enhanced area term in Lemma~\ref{Lemma_E_F=m=0_d=1} in the same way 
as in Lemma~\ref{Lemma_-E_F<-m} and~\ref{Lemma_E_F>m} by attributing the additional factor of $2$ from \eqref{E_F=m=0_U-Rechnung} to the geometric factor $\Phi(\Omega,0,0)$ in the second equality of \eqref{W1-00}.
\end{rem}

Next we consider the higher-dimensional case. Here we want to show that the discontinuity at the origin does not alter the situation in Section \ref{subsec:area}. We proceed as in Section \ref{subsec:area} and apply \cite[Lemma 4.19]{BM-Widom}. The difference is the estimate for the Schatten-von Neumann norm, as the result from the previous paper \cite[Lemma 3.3]{BM-Widom} requires a smooth symbol and therefore does not apply now. Instead we prove a separate Schatten-von Neumann estimate for appropriate symbols, whose proof builds upon the estimates in \cite{sobolevschatten}. A related estimate was recently obtained in \cite[Lemma 4.3]{FinsterSob}. Both estimates were obtained independently of one another.
\begin{lem}\label{Schatten-Estimate}
Let $\Omega$ be a bounded admissible domain and $E_{F}=m=0$. Then, there exists a constant $C>0$, which is independent of $L$ and $b$,
such that for every $L\geq 2$ and $q\in\,]0,1]$
\begin{align}
\big\|\mathbf{1}_\Omega \Op_L \big(\chi_{0}^{(b)}(\mathcal{D})\big)\mathbf{1}_{\Omega^c} \big\|_q^q \leq C \LEFTRIGHT\{.{\begin{array}{ll}
     		 \log [(b+1)L] \ & \text{if } d=1, \\[1ex]
					\displaystyle
     			 [(b+1)L]^{d-1}
							& \text{if } 	\  d\geq 2.
   		 \end{array}}
\end{align}
\end{lem}
\begin{proof}
Let $L\geq 2$ and define the function $\tau:\R^d\rightarrow\R_+$ by 
\begin{equation}
\label{tau-xi-def}
\tau(\xi):=\frac{1}{4}\sqrt{\frac{1}{L^2}+|\xi|^2}.
\end{equation}
It is clearly Lipschitz continuous on $\R^d$ with Lipschitz constant $\frac{1}{4}$. By \cite[Thm.1.4.10]{Hoermander} this function gives rise to a sequence of centres $(\xi_j)_{j\in\N}\subset \R^d$ such that the balls of radius $\tau_j:=\tau(\xi_j)$ about the points $\xi_j$ cover $\R^d$, i.e. $\bigcup_{j\in\N}B_{\tau_{j}}(\xi_j)=\R^d$ and such that at most $N<\infty$ balls intersect in any given point, where the number $N$ only depends on the dimension $d$. We refer to this last property as the finite-intersection property. Furthermore by \cite[Thm.1.4.10]{Hoermander}, there is a partition of unity $(\psi_j)_{j\in\N}$ subordinate to this covering, and for every multi-index $\alpha\in\N_{0}^d$ there exists a constant $\tilde C_{|\alpha|}>0$ such that  
\begin{align}\label{Psi-Derivative}
\sup_{\;j\in\N^{\phantom{d}}} \!\sup_{\xi\in\R^d}|\partial^{\alpha}_{\xi}\psi_j(\xi)|\leq \tilde C_{|\alpha|}\tau(\xi)^{-|\alpha|}.
\end{align}
Here, $|\alpha| := \sum_{k=1}^{d}\alpha_{k}$. We note the (trivial) fact that $(\psi_j)_{j\in\N}$ is independent of $b$. The Lipschitz continuity of $\tau$ guarantees that for every $j\in \N$ and every $\xi\in B_{\tau_j}(\xi_j)$ the following inequality holds
\begin{align}\label{Lipschitz-Property}
\frac{4}{5}\tau(\xi)\leq \tau_j \leq \frac{4}{3} \tau(\xi).
\end{align}
By the definition of $\chi_{0}^{(b)}$ and \eqref{Dirac-Application-Function}, we have 
\begin{equation}
	\label{chi-b-supp}
	\supp\big(\chi_{0}^{(b)}(\mathcal{D})\big) \subset B_{b+1}(0).
\end{equation}
Therefore, there is a finite index set $J\subset\N$, which depends on $b$, such that $\bigcup_{j\in J}B_{\tau_{j}}(\xi_j)\supseteq \supp( \chi_{0}^{(b)}(\mathcal{D}))$. We divide $J$ into two (finite) parts
\begin{align}
J_1:=\{j\in J : B_{\tau_{j}}(\xi_j)\cap \{0\}\neq \emptyset\} \qquad \text{and} \qquad J_2:= J\setminus J_1.
\end{align}
The finite-intersection property implies the $b$-independent upper bound $|J_{1}| \le N$.
Using the $q$-triangle inequality, we estimate
\begin{multline}
	\label{J1-J2-decomp}
	\big\|\mathbf{1}_\Omega \Op_L \big(\chi_{0}^{(b)}(\mathcal{D})\big)\mathbf{1}_{\Omega^c} \big\|_q^q  \\
	\leq \sum_{j\in J_1} \big\|\mathbf{1}_\Omega \Op_L \big(\psi_j\chi_{0}^{(b)}(\mathcal{D})\big)\mathbf{1}_{\Omega^c} 
		\big\|_q^q 
	+ \sum_{j\in J_2} \big\|\mathbf{1}_\Omega \Op_L \big(\psi_j\chi_{0}^{(b)}(\mathcal{D})\big)\mathbf{1}_{\Omega^c} 
		\big\|_q^q.
\end{multline}
If $j\in J_1$ we estimate 
\begin{equation}\label{Estimate-J1}
\big\|\mathbf{1}_\Omega \Op_L \big(\psi_j\chi_{0}^{(b)}(\mathcal{D})\big)\mathbf{1}_{\Omega^c} \big\|_q^q 
\leq \big\| \Op_L^l \big(\phi\psi_j\otimes\mathbb{1}_{n_d})\big\|_q^q \;
\big\|\Op_L \big(\chi_{0}^{(b)}(\mathcal{D})\big)\big\|^q ,
\end{equation}
where $\phi\in C_c^\infty(\R^d)$ with $\phi|_{\Omega}=1$ is supported in the ball $B_{r_1}(0)$ for some radius $r_1>0$ and $\Op_L^l$ denotes the standard left-quantisation functor, cf.\ \cite[Section 2.2, 2.3]{BM-Widom}. As 
$\big\|\Op_L \big(\chi_{0}^{(b)}(\mathcal{D})\big)\big\| = 1$ and 
\begin{align}\label{Estimate-J1-2}
\big\| \Op_L^l \big(\phi\psi_j\otimes \mathbb{1}_{n_d})\big\|_q^q=n_d\big\| \Op_L^l \big(\phi\psi_j\big)\big\|_q^q,
\end{align}
it remains to estimate the $q$-norm for the scalar-valued operator on the right-hand side of \eqref{Estimate-J1-2}.
As $0 \in B_{\tau_j}(\xi_j)$ for all $j\in J_1$, property \eqref{Lipschitz-Property} guarantees that $\phi\psi_j$ is compactly supported in $B_{r_1}(0)\times B_{\frac{1}{3L}}(\xi_j)$. Therefore, by an  application of \cite[Thm 3.1]{sobolevschatten}, the right-hand side of \eqref{Estimate-J1-2} is bounded from above by a constant $C_{1}$ independently of $L$ and $b$. This, in turn, provides the bound
\begin{equation}
	\label{J1-sum}
	\sum_{j\in J_1} \big\|\mathbf{1}_\Omega \Op_L \big(\psi_j\chi_{0}^{(b)}(\mathcal{D})\big)\mathbf{1}_{\Omega^c} 
		\big\|_q^q 
	\le N n_{d} C_{1},
\end{equation}
where the right-hand side is independent of $L$ and $b$.

If instead $j\in J_2$ we note that the symbol $\psi_j\chi_{0}^{(b)}(\mathcal{D})$ is smooth with $\supp \psi_j\chi_{0}^{(b)}(\mathcal{D}) \subset B_{\tau_j}(\xi_j) \subset B_{4\tau_j}(\xi_j)$. Since $4L\tau_j\geq 1$ we can apply \cite[Lemma 3.3, Remark 3.4]{BM-Widom} with $A=\phi\psi_j\chi_{0}^{(b)}(\mathcal{D})$, $t=l$, $s=\max\{r_1,1\}$ and $\tau=4\tau_j$ in combination with \eqref{Psi-Derivative} to obtain
\begin{equation}
\big\|\mathbf{1}_\Omega \Op_L \big(\psi_j\chi_{0}^{(b)}(\mathcal{D})\big)\mathbf{1}_{\Omega^c} \big\|_q^q = \big\|\mathbf{1}_\Omega \Op_L^l \big(\phi\psi_j\chi_{0}^{(b)}(\mathcal{D})\big)\mathbf{1}_{\Omega^c} \big\|_q^q \leq C_{2} (L\tau_j)^{d-1},
\end{equation}
where the constant $C_{2}>0$ is independent of $j\in J_{2}$, as well as of $L$ and $b$.
Therefore, 
\begin{equation}\label{Estimate-Finite-Intersection}
\sum_{j\in J_2} \big\|\mathbf{1}_\Omega \Op_L \big(\psi_j\chi_{0}^{(b)}(\mathcal{D})\big)\mathbf{1}_{\Omega^c} \big\|_q^q \leq C_{2}L^{d-1} \sum_{j\in J_2} \tau_j^{d-1}\leq C_3L^{d-1}\int_{B_{2(b+1)}(0)} \tau(\xi)^{-1} \d\xi,
\end{equation}
where the constant $C_{3}>0$ is independent of $L$ and $b$. The last inequality relies on the finite-intersection property and \eqref{Lipschitz-Property}, as well as the definition of the index set $J$, \eqref{tau-xi-def} and \eqref{chi-b-supp}.  Using the abbreviation $\rho:= 2(b+1)$ and introducing spherical coordinates, the right-hand side of 
\eqref{Estimate-Finite-Intersection} reads 
\begin{align}
 C_3 |\partial B_{1}(0)| \int_0^{L\rho}  \frac{r^{d-1}}{\sqrt{1+r^2}} \; \d r \leq C_{4} \LEFTRIGHT\{.{\begin{array}{@{\,}ll}
     		 \log (\rho L) \ & \text{if } d=1, \\[1ex]
					\displaystyle
     			 (\rho L)^{d-1}
							& \text{if } 	\  d\geq 2,
   		 \end{array}}
\end{align}
for some constant $C_{4} >0$ that is independent of $L$ and $b$.
This estimate, together with \eqref{Estimate-Finite-Intersection}, \eqref{J1-sum} and \eqref{J1-J2-decomp} 
concludes the proof due to $(b+1)L \ge 2$.
\end{proof}

\noindent
With this, the desired estimate follows immediately.

\begin{lem}\label{Lemma_E_F=m=0_d>1}
Let $h$ be as in Theorem \ref{Dirac}, $\Omega$ be a bounded admissible domain, $b\geq 0$, $E_F=m=0$ and $d\geq 2$. Then, for every $L\geq 2$, we have
\begin{align}
\Big\| h\Big(\mathbf{1}_\Omega \Op_L\big(\chi_{0}^{(b)}(\mathcal{D})\big)\mathbf{1}_\Omega\Big)-\mathbf{1}_\Omega \Op_L\Big(h\big(\chi_{0}^{(b)}(\mathcal{D})\big)\Big)\mathbf{1}_\Omega\Big\|_1 \leq CL^{d-1},
\end{align}
where the constant $C>0$ is independent of $L$.
\end{lem}
\begin{proof}
As $\chi_{0}^{(b)}$ and thus the operator $\Op_L\big(\chi_{0}^{(b)}(\mathcal{D})\big)$ are bounded, we can assume that the function $h$ is compactly supported. An application of \cite[Lemma 4.19]{BM-Widom} with $A=\Op_L\big(\chi_{0}^{(b)}(\mathcal{D})\big)$, $P=\mathbf{1}_\Omega$ and arbitrary 
$q\in \,]0,\gamma[\,$ yields
\begin{align}
\Big\|h\Big(\mathbf{1}_\Omega \Op_L\big(\chi_{0}^{(b)}(\mathcal{D})\big)\mathbf{1}_\Omega\Big)-\mathbf{1}_\Omega \Op_L\Big(h\big(\chi_{0}^{(b)}(\mathcal{D})\big)\Big)\mathbf{1}_\Omega\Big\|_1
\leq C_1 \Big\|\mathbf{1}_\Omega \Op_L \big(\chi_{0}^{(b)}(\mathcal{D})\big)\mathbf{1}_{\Omega^c} \Big\|_q^q,
\end{align}
where the constant $C_1>0$ does not depend on the parameter $L$. Therefore, the claim follows from an
application of Lemma \ref{Schatten-Estimate} in the case $d\geq 2$.
\end{proof}

%
\subsection{Proof of the main result}
\label{subsec:proof}

All that remains to be done is to collect the previous statements.  

\begin{proof}[Proof of Theorem~\ref{Dirac}]
For fixed $L>0$ we rewrite the left-hand side of \eqref{Dirac-stat} as
\begin{multline}\label{Rewrite-Dirac}
		\frac{1}{n_{d}}\tr_{L^2(\R^d)\otimes\C^{n_{d}}} \!
			\bigg[h\Big(\mathbf{1}_{\Lambda}\Op_{L}\big(\chi_{E_F}^{(b)}(\mathcal{D})\big)\mathbf{1}_{\Lambda}\Big) 
			\\+ h\Big(\mathbf{1}_{\Lambda'\setminus\Lambda}\Op_{L}\big(\chi_{E_F}^{(b)}(\mathcal{D})\big)
				\mathbf{1}_{\Lambda'\setminus\Lambda}\Big)
			- h\Big(\mathbf{1}_{\Lambda'}\Op_{L}\big(\chi_{E_F}^{(b)}(\mathcal{D})\big)\mathbf{1}_{\Lambda'}\Big)\bigg] 
	\end{multline}
by \eqref{Rewrite-PDO} and \eqref{Rewrite-L-PDO}. As both $\Lambda$ and $\Lambda'$ are bounded piece-wise $C^1$-admissible domains, the domain $\Lambda'\setminus \Lambda$ is also a bounded piece-wise $C^1$-admissible domain up to a set of measure zero. 

Firstly, we treat the situations where there will be a logarithmic enhancement. 
So, we consider either the case $|E_F|>m$ or the case $E_F=m=0$ in $d=1$. 
In order to cover these simultaneously, we introduce the asymptotic volume coefficient 
\begin{equation}
	V(h,b,E_{F},m):= \LEFTRIGHT\{.{\begin{array}{@{\,}ll}V_{-}(h,b,E_{F},m), &  \text{if~~} E_{F} < -m, \\
				V_{+}(h,b,E_{F},m), &  \text{if~~} E_{F} > m, \\ V_{0}(h,b), & \text{if~~} E_{F}=m=0 \text{~~and~~} d=1.
			\end{array}}
\end{equation}
In the first case, we apply Lemma \ref{Lemma_-E_F<-m}, in the second case Lemma~\ref{Lemma_E_F>m}
and in the third case Lemma~\ref{Lemma_E_F=m=0_d=1} to all three terms in \eqref{Rewrite-Dirac}. 
This yields the asymptotic expansion 
\begin{multline}
	\label{almost-done}
	\frac12\, L^d V(h,b,E_{F},m) \Big(|\Lambda |+|\Lambda'\setminus\Lambda|-|\Lambda'|\Big)\\
	+ \frac{1}{2}\, L^{d-1}\log L  \Big( W(h,\Lambda,E_F,m) + W(h,\Lambda'\setminus\Lambda,E_F,m) 
			- W(h,\Lambda',E_F,m)\Big)
	+ o(L^{d-1}\log L)
\end{multline}
of \eqref{Rewrite-Dirac}, as $L\rightarrow\infty$.
Now, the volumes in the first line of \eqref{almost-done} add up to zero,
and by the definition of $W(h,\Lambda,E_F,m)$ in Theorem~\ref{Dirac} we have 
\begin{equation}
	W(h,\Lambda,E_F,m) + W(h,\Lambda'\setminus\Lambda,E_F,m) - W(h,\Lambda',E_F,m)
	= 2W(h,\Lambda,E_F,m).
\end{equation}
Thus, we infer the desired result 
\begin{multline}
		\frac{1}{n_{d}}\tr_{L^2(\R^d)\otimes\C^{n_{d}}} \!
			\Big[h\big(\mathbf{1}_{\Lambda_{L}}\chi_{E_F}^{(b)}(D)\mathbf{1}_{\Lambda_{L}}\big) 
			+ h\big(\mathbf{1}_{\Lambda'_{L}\setminus\Lambda_{L}}\chi_{E_F}^{(b)}(D)
				\mathbf{1}_{\Lambda'_{L}\setminus\Lambda_{L}}\big)
			- h\big(\mathbf{1}_{\Lambda'_{L}}\chi_{E_F}^{(b)}(D)\mathbf{1}_{\Lambda'_{L}}\big)\Big] \\ 
		= L^{d-1}\log L \  W(h,\Lambda,E_F,m) + o(L^{d-1}\log L),
\end{multline}
as $L\rightarrow\infty$.

Secondly, we turn to the situations without logarithmic enhancement, that is, 
the case $|E_F|\leq m \neq 0$ and the case 
$E_{F} = m =0$ in $d\geq 2$. We observe that for fixed $L>0$ we have
\begin{multline}\label{Lambda-Lambda'-Comparison}
\mathbf{1}_{\Lambda} \Op_L\Big(h\big(\chi_{E_F}^{(b)}(\mathcal{D})\big)\Big)\mathbf{1}_{\Lambda}+\mathbf{1}_{\Lambda'\setminus\Lambda} \Op_L\Big(h\big(\chi_{E_F}^{(b)}(\mathcal{D})\big)\Big)\mathbf{1}_{\Lambda'\setminus\Lambda}-\mathbf{1}_{\Lambda'} \Op_L\Big(h\big(\chi_{E_F}^{(b)}(\mathcal{D})\big)\Big)\mathbf{1}_{\Lambda'} \\ = -\mathbf{1}_{\Lambda} \Op_L\Big(h\big(\chi_{E_F}^{(b)}(\mathcal{D})\big)\Big)\mathbf{1}_{\Lambda'\setminus\Lambda}-\mathbf{1}_{\Lambda'\setminus\Lambda}\Op_L\Big(h\big(\chi_{E_F}^{(b)}(\mathcal{D})\big)\Big)\mathbf{1}_{\Lambda}.
\end{multline}
Both operators in the last line of \eqref{Lambda-Lambda'-Comparison} have vanishing trace. Therefore, we can rewrite \eqref{Rewrite-Dirac} as 
\begin{align}\label{Lambda-Lambda'-Reduced}
	\frac{1}{n_{d}}\tr_{L^2(\R^d)\otimes\C^{n_{d}}} \!
		\bigg[&h\Big(\mathbf{1}_{\Lambda}\Op_{L}\big(\chi_{E_F}^{(b)}(\mathcal{D})\big)
				\mathbf{1}_{\Lambda}\Big)
		- \mathbf{1}_{\Lambda} \Op_L\Big(h\big(\chi_{E_F}^{(b)}(\mathcal{D})\big)\Big)
				\mathbf{1}_{\Lambda} \nonumber \\
	& + h\Big(\mathbf{1}_{\Lambda'\setminus\Lambda}\Op_{L}\big(\chi_{E_F}^{(b)}(\mathcal{D})\big)
				\mathbf{1}_{\Lambda'\setminus\Lambda}\Big)
			- \mathbf{1}_{\Lambda'\setminus\Lambda} \Op_L\Big(h\big(\chi_{E_F}^{(b)}(\mathcal{D})
				\big)\Big)\mathbf{1}_{\Lambda'\setminus\Lambda}\nonumber \\ 
	& - h\Big(\mathbf{1}_{\Lambda'}\Op_{L}\big(\chi_{E_F}^{(b)}(\mathcal{D})\big)
				\mathbf{1}_{\Lambda'}\Big)+\mathbf{1}_{\Lambda'} \Op_L\Big(h\big(\chi_{E_F}^{(b)}
				(\mathcal{D})\big)\Big)\mathbf{1}_{\Lambda'}\bigg]
\end{align}
and estimate the trace norm of the operator difference in each of the three lines  
by Lemma~\ref{Lemma_E_F=<m>0} in the case $|E_F|\leq m \neq 0$ and by 
Lemma~\ref{Lemma_E_F=m=0_d>1} in the case $E_{F} = m =0$ and $d\geq 2$. 
This yields the existence of a constant $C>0$ such that the modulus of 
\eqref{Lambda-Lambda'-Reduced} is bounded from above by $CL^{d-1}$ for all $L\geq 2$, 
and the proof is complete.
\end{proof}

\section*{Acknowledgement}
We thank Edgardo Stockmeyer for stimulating and helpful discussions. This work was partially supported by the Deutsche 
Forschungsgemeinschaft (DFG, German Research Foundation) -- TRR 352 ``Mathematics of Many-Body Quantum Systems and Their 
Collective Phenomena" -- Project-ID 470903074.


\IfFileExists{mybib.bst}{\bibliographystyle{mybib}}{\bibliographystyle{../mybib}}
\IfFileExists{references.bib}{\bibliography{references-Ruth,references}}{\bibliography{../references-Ruth,../references}}

\end{document}